\documentclass[aps,prl,twocolumn,groupedaddress]{revtex4}

\usepackage{amsmath, amsthm, amssymb, pstricks}

\newcommand{\newmax}{\operatornamewithlimits{max}}
\newcommand{\e}{\varepsilon}

\newtheorem{thm}{Theorem}
\newtheorem{lem}{Lemma} 
\newtheorem{definition}{Definition}

\begin{document}

\title{Distilling Non-Locality}

\author{Manuel Forster} \author{Severin Winkler} \author{Stefan Wolf}
\affiliation{Computer Science Department, ETH Z\"urich, CH-8092
  Z\"urich, Switzerland}

\date{\today}

\begin{abstract}

  Two parts of an entangled quantum state can have a correlation in
  their joint behavior under measurements that is unexplainable by
  shared classical information. Such correlations are called {\em
    non-local\/} and have proven to be an interesting resource for
  information processing. Since non-local correlations are more useful
  if they are stronger, it is natural to ask whether weak non-locality
  can be amplified. We give an affirmative answer by presenting the
  first protocol for distilling non-locality in the framework of
  generalized non-signaling theories. Our protocol works for both
  quantum and non-quantum correlations. This shows that in many
  contexts, the extent to which a single instance of a correlation can
  violate a CHSH inequality is not a good measure for the usefulness of
  non-locality. A more meaningful measure follows from our results.

\end{abstract}

\pacs{03.67.-a,03.65.Ud}

\maketitle

When two separated parts of a quantum state are measured in fixed bases,
then the outcomes can show a correlation. Whereas this may be surprising
from a physical point of view, it is not from the standpoint of
information: such correlations could be explained by randomness shared
when the two particles were generated.

If one considers, however, different possible measurement settings on
the two sides, then correlations of a stronger kind can arise, which are
{\em unexplainable by shared randomness only\/}~\cite{Bell-1964}: This
is {\em non-locality\/}.

Quantum mechanics is non-local but not maximally so. There are stronger
correlations still in accordance with the non-signaling postulate of
relativity~\cite{PR-1994}. This fact motivated the study of so-called
{\em generalized non-signaling theories\/}~\cite{barrett05,barret-2005}
in which quantum correlations are a special case. Following this general
approach to non-locality, we study correlations between the joint
behavior of the two ends of a bipartite input-output {\em system},
characterized by a conditional probability distribution $P(ab|xy)$. Let
$x$ and $a$ be the input and output on the left-hand side of the system,
and $y$ and $b$ the corresponding values on the right-hand side.
\\
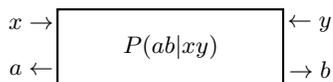
\begin{figure}[h]
  \begin{center}
    \begin{pspicture}(7,1)
      \psline{-}(2,0)(5,0)(5,1)(2,1)(2,0)
      \rput(1.65,0.8){$x\rightarrow$} \rput(1.65,0.2){$a\leftarrow$}
      \rput(5.35,0.8){$\leftarrow y$} \rput(5.35,0.2){$\rightarrow b$}
      \rput(3.5,0.5){$P(ab|xy)$}
    \end{pspicture}
  \end{center}
  \label{bi-partite system1}
\end{figure}

\noindent
We call such a system {\em local\/} if it is explainable by shared
classical information. On the other hand, it is {\em signaling\/} if it
allows for message transmission in either direction.

John Bell has given properties that local systems have, namely certain
inequalities they must obey. Hence, {\em violation\/} of such an
inequality is a witness of {\em non\/}-locality. In the case where both
inputs and both outputs are {\em binary}, the only such inequality (up
to symmetries) is the so-called CHSH (after Clauser, Horne, Shimony,
Holt) inequality~\cite{CHSH-1969}. Furthermore, the set of eight CHSH
inequalities is complete for binary systems in the sense that if none of
them is violated, then the system is local.

In this letter we restrict ourselves to the state space of binary
input/binary output non-signaling systems. We refer to
\cite{barret-2005} for a detailed description of this set.

Non-local correlations are not only a fascinating phenomenon, but have
as well been shown to be an interesting resource for information
processing. Examples are device-independent secrecy of quantum
cryptography~\cite{Hardy05} and non-local
computation~\cite{Linden-2006}. Furthermore, the existence of
non-locality that is super-quantum to some extent would have dramatic
consequences on communication complexity~\cite{brassard-2006}. This
extends the fact that {\em maximal\/} non-locality would collapse
communication complexity, i.e., allows to compute every distributed
Boolean function with just one communicated bit~\cite{vandam-2005}.

The extent by which a Bell inequality, e.g., CHSH, is violated can be
taken as a measure for non-locality. Not surprisingly, non-locality is
a more useful resource, the stronger it is. For instance, the violation
of CHSH gives a lower bound to the uncertainty of a third party about
the output bits of a non-signaling system, which is better the stronger
the violation is.

Motivated by these facts, we study the problem of whether non-locality
can be amplified: Can stronger non-locality be obtained from a number of
weakly non-local systems? We consider protocols for non-locality
distillation executed by two parties having access to weakly non-local
systems. The parties on the two sides can carry out arbitrary operations
on their pieces of information, but they {\em cannot\/} communicate.

Note that such protocols should not be confused with protocols for {\em
  entanglement\/} distillation: There, the input and output are (weakly
and strongly, respectively) entangled quantum states, and the allowed
operations are classical communication and local quantum operations. The
existence of certain entanglement distillation protocols {\em without
  communication\/} is known~\cite{bennett-1996}, but this result is
independent of ours.

There are several known impossibility results on non-locality
distillation. First, it is not possible to create non-locality from
locality, i.e., to pass the Bell bound \cite{Bell-1964}. Second, there
exists no non-locality distillation which can pass the Tsirelson bound
\cite{cirelson-1980} if the non-local systems can be simulated by
quantum mechanics. Third, a simple inductive argument shows that a
system that exhibits the algebraically maximal possible CHSH violation
cannot be obtained from weaker ones. Fourth, it has been shown recently
that the CHSH violation of two copies of isotropic systems cannot be
distilled~\cite{short08}. And finally, it has been proven
in~\cite{dukaric-2008} that there exists an infinite number of isotropic
systems for which non-locality distillation cannot be achieved.


An open question which remains is whether non-locality can be distilled
at all. We answer this question affirmatively.
\medskip\\
\noindent {\bf Main Result.}
There exists a protocol which allows the distillation of certain, both
quantum-mechanically achievable and unachievable, binary non-local
systems.

\section{Definitions}

A binary input-output system characterized by a conditional probability
distribution $P(ab|xy)$ is \emph{non-signaling\/} if one cannot signal
from one side to the other by the choice of the input. This means that
the marginal probabilities $P(a|x)$ and $P(b|y)$ are independent of $y$
and $x$, respectively, i.e.,

\begin{align*}
  \sum_bP(ab|xy)&=\sum_bP(ab|xy')\equiv P(a|x)\ \forall a,x,y,y',\\
  \sum_aP(ab|xy)&=\sum_aP(ab|x'y)\equiv P(b|y)\ \forall b,x,x',y.
\end{align*}

When using a non-signaling system, a party receives its output
immediately after giving its input, independently of whether the other
has given its input already. This prevents the parties from signaling by
delaying their inputs.

If appropriate we represent a system by its probability distribution
$P(ab|xy)$ in matrix notation as
\[\left[
  \begin{array}{cccc}
    P(00|00)&P(01|00)&P(10|00)&P(11|00)\\
    P(00|01)&P(01|01)&P(10|01)&P(11|01)\\
    P(00|10)&P(01|10)&P(10|10)&P(11|10)\\
    P(00|11)&P(01|11)&P(10|11)&P(11|11)\\
  \end{array}\right].
\]

Given $P(ab|xy)$ ($P$) we define the set of four correlation functions:
\begin{align*}
  X_{xy}(P)=P(00|xy)+P(11|xy)-P(01|xy)-P(10|xy),
\end{align*}
for $xy=00,01,10,11$. The corresponding system is local if and only if
its correlation functions satisfy the following CHSH inequalities
\cite{CHSH-1969}:

\begin{align}
  |X_{xy}(P)+X_{x\bar{y}}(P)+X_{\bar{x}y}(P)-X_{\bar{x}\bar{y}}(P)|\leq
  2,\label{chsh}
\end{align}
for $xy=00,01,10,11$. (We use $\bar{x}$ and $\bar{y}$ to indicate bit
flips, that is, $\bar{0}=1$ and $\bar{1}=0$.)

In order to measure the non-locality of a system we will use the maximal
violation of a CHSH inequality:

\begin{definition}
  We define the CHSH non-locality of a binary input, binary output
  system $P$ as
  \[
  NL[P]:=\newmax_{xy}|X_{xy}(P)+X_{x\bar{y}}(P)+X_{\bar{x}y}(P)-X_{\bar{x}\bar{y}}(P)|,
  \]
  Note that $NL[P]>2$ indicates that the correlation $P$ violates CHSH
  and is therefore called non-local.
\end{definition}

Quantum mechanics predicts violations of the CHSH inequalities
(\ref{chsh}) up to $2\sqrt{2}$. However, this bound is only
necessary. The necessary \emph{and} sufficient condition for a set of
four numbers to be reached by quantum mechanics was found by Landau
\cite{landau1988} and Tsirelson \cite{Tsirelson93} (see also Masanes
\cite{masanes-2003}).

\begin{lem}\label{quant_def}
  \noindent A set of correlation functions $X_{xy}$, $xy=00,01,10,11$,
  can be reached by a quantum state and some local observables if and
  only if they satisfy the following four inequalities:
  \begin{align*}
    |\arcsin X_{xy}+\arcsin X_{x\bar{y}}+\arcsin X_{\bar{x}y}-\arcsin
    X_{\bar{x}\bar{y}}| \leq \pi.
  \end{align*}
\end{lem}

Using the terms introduced above we formally define a non-locality
distillation protocol as follows:

\begin{definition}
  A non-locality-distillation protocol is executed by two parties (Alice
  and Bob) without communication. It simulates a binary input/binary
  output system $P^n$ by classical (local) operations on $n$ non-local
  resource systems $P$, such that $NL[P^n]>NL[P]>2$.
\end{definition}

\section{Results}

In the following we present a non-locality-distillation protocol and
distillable non-local resource systems. We will also present resource
systems that are measurable on a quantum state and can be used by our
protocol to distill (quantum) non-locality.

We define the protocol $\text{NDP}_n(P)$ on $n$ non-signaling systems
$P$ between Alice and Bob as follows: On inputs $x$ to Alice and $y$ to
Bob the parties input $x$ and $y$ to all $n$ systems in parallel and
receive outputs $(a_1,\dots,a_n)$ and $(b_1,\dots,b_n)$,
respectively. The parties then locally compute their output bits as
$a=\sum_{i=1}^na_i$ (mod 2) for Alice and $b=\sum_{i=1}^nb_i$ (mod 2)
for Bob. The whole protocol is illustrated in more detail in Figure
\ref{dist prot}.

\begin{figure}[ht]
  \begin{center}
    \begin{tabular}{lcr}
      ALICE&\hspace{-0.5cm}$\text{NDP}_n(P)$&\hspace{-0.2cm}BOB\\\hline\hline
      $x\in\{0,1\}$&\hspace{-0.2cm}inputs&\hspace{-0.2cm}$y\in\{0,1\}$\\
      &&\\
      $a_1$&\hspace{-0.2cm}$P(a_1b_1|xy)$&\hspace{-0.2cm}$b_1$\\
      $a_2$&\hspace{-0.2cm}$P(a_2b_2|xy)$&\hspace{-0.2cm}$b_2$\\
      $\vdots$&\hspace{-0.2cm}$\vdots$&\hspace{-0.2cm}$\vdots$\\
      $a_n$&\hspace{-0.2cm}$P(a_nb_n|xy)$&\hspace{-0.2cm}$b_n$\\
      &&\\
      $a=\sum_{i=1}^na_i\mod 2$&\hspace{-0.2cm}outputs&\hspace{-0.2cm}$b=\sum_{i=1}^nb_i\mod 2$
    \end{tabular}
  \end{center}
  \caption{The final outputs are a simple exclusive-or of all the
    outputs obtained from a parallel usage of the available non-local
    resource systems.\label{dist prot}}
\end{figure}

For $0<\e\leq 1$ we define the following non-signaling system
\[
P_\varepsilon=\left[
  \begin{array}{cccc}
    1/2&0&0&1/2\\
    1/2&0&0&1/2\\
    1/2&0&0&1/2\\
    1/2-\varepsilon/2&\varepsilon/2&\varepsilon/2&1/2-\varepsilon/2
  \end{array}\right]
\]
as our non-local distillation resource with CHSH non-locality
$NL[P_\e]=3-(1-2\e)>2$. With probability $\e$ this system behaves like a
PR-box~\cite{PR-1994} and with probability $1-\e$ it outputs perfectly
correlated random bits.

\begin{thm}\label{thm1}
  For $n>1$ and $0<\varepsilon<1/2$ the protocol $\text{NDP}_n(P_\e)$ is
  a non-locality-distillation protocol.
\end{thm}

\begin{proof}[Proof of Theorem \ref{thm1}]
  Obviously, $\text{NDP}_n(P_\e)$ describes only classical, local
  operations on Alice's and Bob's side. Furthermore,
  $\text{NDP}_n(P_\e)$ simulates another binary input/binary output
  system $P_\e^n$ with CHSH non-locality
  \begin{align*}
    NL[P_\e^n]&
    =X_{00}(P_\e^n)+X_{01}(P_\e^n)+X_{10}(P_\e^n)-X_{11}(P_\e^n)\\
    &=3-X_{11}(P_\e^n)\\
    &=3-(P_\e^n(00|11)+P_\e^n(11|11)\\
    &~~~~~~~~~~-P_\e^n(01|11)-P_\e^n(10|11)).
  \end{align*}
  Here, we used that $X_{00}(P_\e^n),X_{01}(P_\e^n),X_{10}(P_\e^n)$ are
  constant functions reaching the algebraic maximum of 1. Analogously to
  $P_\e^n$, let $P_\e^{n-1}$ denote the system simulated by
  $\text{NDP}_{n-1}(P_\e)$. Using
  \begin{align*}
    P_\e^n(00|11)&=P_\e^n(11|11)\\
    &=(1/2-\e/2)(P_\e^{n-1}(00|11)+P_\e^{n-1}(11|11))\\
    &~~~+\e/2(P_\e^{n-1}(01|11)+P_\e^{n-1}(10|11)),
  \end{align*}
  \begin{align*}
    P_\e^n(01|11)&=P_\e^n(10|11)\\
    &=\e/2(P_\e^{n-1}(00|11)+P_\e^{n-1}(11|11))\\
    &~~~+(1/2-\e/2)(P_\e^{n-1}(01|11)+P_\e^{n-1}(10|11))
  \end{align*}
  we derive
  \begin{align*}
    NL[P_\e^n]&=3-(1-2\e)(P_\e^{n-1}(00|11)+P_\e^{n-1}(11|11)\\
    &~~~-P_\e^{n-1}(01|11)-P_\e^{n-1}(10|11))
    \\
    &=3-(1-2\e)X_{11}(P_\e^{n-1}).
  \end{align*}
  Therefore, we have established
  \begin{align*}
    NL[P_\e^n]&=3-X_{11}(P_\e^n)=3-(1-2\e)X_{11}(P_\e^{n-1})\\
    &=3-(1-2\e)^{n-1}X_{11}(P_\e)=3-(1-2\e)^n.
  \end{align*}
  For $0<\e<1/2$ we can guarantee $3-(1-2\e)^n>3-(1-2\e)^{n-1}$, which
  implies $NL[P_\e^n]>NL[P_\e]$.\end{proof}

In the limit we have
$\lim_{n\rightarrow\infty}NL[P_\e^n]=\lim_{n\rightarrow\infty}3-(1-2\e)^n=3$.

Note that the presented systems are not quantum-physically
realizable. This allows our protocol to pass the Tsirelson bound using
$P_\e$ with $0<\e\leq \sqrt{2}-1$ as resource systems. In the following
we show that non-locality distillation is also possible for systems
available in quantum mechanics. We therefore introduce a more general
parameterized system (positivity is ensured by $0\leq\e,\delta\leq 1$):
\[P_{\varepsilon,\delta}=\left[
  \begin{array}{cccc}
    1/2-\delta/2&\delta/2&\delta/2&1/2-\delta/2\\
    1/2-\delta/2&\delta/2&\delta/2&1/2-\delta/2\\
    1/2-\delta/2&\delta/2&\delta/2&1/2-\delta/2\\
    1/2-\varepsilon/2&\varepsilon/2&\varepsilon/2&1/2-\varepsilon/2
  \end{array}\right]
\]
This system has CHSH non-locality $3(1-2\delta)-(1-2\e)$. For $\delta=0$
we have $P_{\e,\delta}=P_\e$.

Note that we have chosen the two example resource systems because of
their simplicity. This should not suggest that these exact systems are
the only systems distillable by our protocol. Obviously the
distillability of a system with the presented protocol does only depend on its
correlation functions and not on the marginals.

\begin{thm}\label{thm2}
  There exist $0<\delta<\e<1/2$ and $n>1$ such that
  $P_{\varepsilon,\delta}$ is a quantum system and
  $\text{NDP}_n(P_{\varepsilon,\delta})$ is a non-locality-distillation
  protocol.
\end{thm}

\begin{proof}[Proof of Theorem \ref{thm2}]
  Protocol $\text{NDP}_n(P_{\e,\delta})$ simulates another two input/two
  output system $P_{\e,\delta}^n$. By setting $\delta<\e$ and following
  a similar reasoning as in the proof of Theorem \ref{thm1} we obtain
  \begin{align*}
    NL[P_{\e,\delta}^n]&=X_{00}(P_{\e,\delta}^n)+X_{01}(P_{\e,\delta}^n)+X_{10}(P_{\e,\delta}^n)-X_{11}(P_{\e,\delta}^n)
    \\
    &=3(1-2\delta)^n-(1-2\e)^n.
  \end{align*}
  We can find values $n$ and $0<\delta<\e<1/2$ (for example,
  $n=2,\e=0.01,\delta=0.002$) such that $P_{\e,\delta}$ is at the same
  time distillable, i.e.,
  \[3(1-2\delta)^n-(1-2\e)^n>3(1-2\delta)-(1-2\e)\] and a quantum
  system, i.e.,
  \[
  \begin{array}{l}
    |3\arcsin (1-2\delta)-\arcsin (1-2\e)| \leq \pi,\\
    |\arcsin (1-2\delta)+\arcsin (1-2\e)| \leq \pi.
  \end{array}
  \]
  Lemma \ref{quant_def} only guarantees that the correlation functions
  of $P_{\e,\delta}$ are obtainable by quantum mechanics. But Alice and
  Bob can make their outputs locally uniform such that the correlation
  functions are preserved using shared randomness. Thus $P_{\e,\delta}$
  is a quantum system if its correlation functions are obtainable by
  quantum mechanics.

  Therefore, we can achieve $NL[P_{\e,\delta}^n]>NL[P_{\e,\delta}]$,
  which means that non-locality has been distilled with quantum systems
  as resources.
\end{proof}

\noindent A natural follow up question concerns the \emph{maximum}
non-locality our protocol can distill using the quantum systems
presented above.

Optimal parameters $n,\e,\delta$ maximize the term
$NL[P_{\e,\delta}^n]=3(1-2\delta)^n-(1-2\e)^n$ with respect to the
conditions that $NL[P_{\e,\delta}^n]>NL[P_{\e,\delta}]$ and that
$P_{\e,\delta}$ is a quantum system (Lemma \ref{quant_def}). The maximal
non-locality that can be distilled by $\text{NDP}_n(P_{\e,\delta})$ is
\[NL[P_{\e_{\text{max}},\delta_{\text{max}}}^{n_\text{max}}]=1+\sqrt{2},\]
where $n_{\text{max}}=2$, $\e_{\text{max}}\simeq 0.30866\text{ and
}\delta_{\text{max}}\simeq 0.03806$.

\section{A new measure of non-locality}

The possibility of distillation motivates the definition of a new
measure for non-locality, namely the maximal CHSH violation {\em
  achievable from many realizations of a given system by any
  distillation protocol}.

As an example application consider the computation of the non-locally
distributed version of the AND function: Two separated parties are given
inputs $x_1,x_2$ and $y_1,y_2$, respectively and have to find outputs
$a$ and $b$, such that the probability of obtaining
\begin{align}
  a\oplus b=(x_1\oplus y_1)\wedge(x_2\oplus y_2)\label{and}
\end{align}
is maximal. Quantum mechanics allows no advantage over the optimal,
classical strategy~\cite{Linden-2006}. Rearranging (\ref{and}) yields a
strategy with success probability directly related to the CHSH violation
of a given resource system. By non-locality distillation of copies of
our arbitrarily weak non-local system $P_\e$ a higher success
probability above the quantum bound can be reached. This illustrates
that distillable systems like $P_\e$ -- although located arbitrarily
``close'' to the quantum bound -- are a stronger computational resource
than any quantum system. Therefore, we obtain a separation of quantum
and post-quantum correlations below the Tsirelson bound in terms of
information processing power.

\section{Conclusion}

We have shown that non-locality of binary-input binary-output systems,
measured by how strongly the CHSH inequality is violated, can be
amplified. More precisely, we have shown that certain systems which
violate CHSH arbitrarily weakly (achieving the value $2+2\varepsilon$),
but that are nevertheless {\em not\/} realizable by quantum physics, can
be distilled.

Furthermore, we show that even certain quantum-mechanically achievable
systems can be distilled: Interestingly, the achievable limit by our
protocol is then the exact mean ($1+\sqrt{2}$) between the classical
($2$) and the quantum ($2\sqrt{2}$) bounds.

Our result complements previous ones, stating that the distillability of
non-locality of two {\em isotropic\/} systems is
impossible~\cite{short08} and at most very limited in
general~\cite{dukaric-2008}. Isotropic systems are an important special
case because they are the worst case with respect to distillability,
i.e., every non-signaling system can be turned into an isotropic system
such that non-locality is preserved using shared randomness only (this
transformation is known as
\emph{depolarization}~\cite{masanes-2005}). Therefore, these
non-distillable isotropic systems cannot be used to simulate the
distillable resources defined here. In other words, bipartite isotropic
and non-isotropic non-signaling (and quantum) systems are in general
inequivalent correlations, although they exhibit the same violation of
CHSH.

The possibility of distillation motivates the definition of a new
measure for non-locality. Clearly, this measure is significant in any
context where non-locality is used as a resource for information
processing, and where the number of realizations available is not
limited to one.

\begin{acknowledgments}
  We thank Dejan D. Dukaric and Esther H\"anggi for useful
  discussions. This work was funded by the Swiss National Science
  Foundation (SNSF).
\end{acknowledgments}

\end{document}